\pdfoutput = 1

\documentclass[letterpaper, 11pt]{article}

\usepackage{amsmath, amssymb, amsfonts, amsthm}
\usepackage{relsize}
\usepackage{algorithm}
\usepackage{algpseudocode}
\usepackage[margin = 1 in]{geometry}
\usepackage{url}
\usepackage{hyperref}

\newtheorem{definition}{Definition}
\newtheorem{theorem}{Theorem}
\newtheorem{lemma}{Lemma}

\newtheorem{corollary}{Corollary}

\theoremstyle{remark}
\newtheorem{remark}{Remark}

\newcommand{\defeq}{\stackrel{\mathsmaller{\mathsf{def}}}{=}}

\newcommand{\E}{\text{E}}

\renewcommand{\geq}{\geqslant}
\renewcommand{\ge}{\geqslant}
\renewcommand{\leq}{\leqslant}
\renewcommand{\le}{\leqslant}

\newcommand{\lb}{\left}
\newcommand{\rb}{\right}

\newcommand{\shortestPath}{\textsf{shortestPath}}

\newcommand{\boundary}{\mathcal{D}}
\newcommand{\good}{\textsf{Good}}
\newcommand{\goodtree}{\textsf{GoodTL}}
\newcommand{\treelike}{\textsf{TreeLike}}
\newcommand{\byz}{\textsf{Byz}}
\newcommand{\ball}{\mathcal{B}}
\newcommand{\diam}{\text{diam}}
\def\polylog{\operatorname{polylog}}
\def\poly{\operatorname{poly}}

\date{}

\title{Byzantine-Resilient Counting in Networks}

\author{Soumyottam Chatterjee \and Gopal Pandurangan \and Peter Robinson}

\begin{document}

\maketitle
\thispagestyle{empty}
\begin{abstract}

We present two distributed algorithms for the {\em Byzantine counting problem}, which is concerned with estimating the size of a network in the presence of a large number of Byzantine nodes.

In an $n$-node network ($n$ is unknown), our first algorithm, which is {\em deterministic}, finishes in $O(\log{n})$ rounds  and is time-optimal. This algorithm can tolerate up to $O(n^{1 - \gamma})$ arbitrarily (adversarially) placed Byzantine nodes for any arbitrarily small (but fixed) positive constant $\gamma$. It outputs a (fixed) constant factor estimate of $\log{n}$ that would be known to all but $o(1)$ fraction of the good nodes. This algorithm works for \emph{any} bounded degree expander network. However, this algorithms assumes that good nodes can send arbitrarily large-sized messages in a round.

Our second algorithm is {\em randomized} and most good nodes send only small-sized messages.\footnote{Throughout this paper, a small-sized message is defined to be one that contains $O(\log{n})$ bits in addition to at most a constant number of node IDs.} This algorithm  works in \emph{almost all} $d$-regular  graphs. It tolerates up to $B(n) = n^{\frac{1}{2} - \xi}$  (note that $n$ and $B(n)$ are unknown to the algorithm) arbitrarily (adversarially) placed Byzantine nodes, where $\xi$ is any arbitrarily small (but fixed) positive constant. This algorithm takes $O(B(n)\log^2{n})$ rounds and outputs a (fixed) constant factor estimate of $\log{n}$ with probability at least $1 - o(1)$. The said estimate is known to most nodes, i.e., $\geq  (1 - \beta)n$ nodes for any arbitrarily small (but fixed) positive constant $\beta$.

To complement our algorithms, we also present an impossibility result that shows that it is impossible to estimate the network size with any reasonable approximation with any non-trivial probability of success if the network does not have sufficient vertex expansion.

Both algorithms are the first such algorithms that solve Byzantine counting in sparse, bounded degree networks under very general assumptions. Both algorithms are fully local and need no global knowledge. Our algorithms can be used for the design of efficient distributed algorithms resilient against Byzantine failures, where the knowledge of the network size --- a global parameter --- may not be known a priori.\\


\end{abstract}

\clearpage

\setcounter{page}{1}


\section{Introduction} \label{sec:intro}

The recent surge in the popularity of decentralized peer-to-peer protocols has renewed the interest in achieving Byzantine fault-tolerance in sparse networks of untrusted participants. In this work, we study the fundamental problem of {\em Byzantine counting} where the goal is to estimate the number of nodes in a network in the presence of a large number of Byzantine nodes. We say that a node $u$ is \emph{good} or \emph{honest} if $u$ is not a Byzantine node. We assume that the Byzantine nodes are {\em arbitrarily} distributed in the network and that when a Byzantine node sends a message over an edge, it cannot fake its ID. We note that both of these assumptions are quite typical in the literature \cite{Dwork_1988, Upfal_1994, King_2006_FOCS, Augustine_2012, Augustine_2013_PODC, Augustine_2015_DISC}. Please refer to Section \ref{sec:model} for more details on the distributed computing model.

We focus on the Byzantine counting problem in the context of \emph{sparse} networks because of the following reasons.
\begin{enumerate}
    \item Peer-to-peer networks and most other large-scale, real-world networks happen to be sparse.
    \item In a $d$-regular network, if the degree $d$ is a non-constant function of $n$, e.g., if $d = \Theta(\log{n})$, then it might become trivial for a node to estimate $n$ from its knowledge of its own degree $d$.
\end{enumerate}

Essentially all known algorithms studied in the literature for solving problems like \emph{Byzantine consensus} and \emph{Byzantine leader election} in sparse networks require an underlying {\em expander graph}: the expansion property is \emph{needed} in tolerating a large number of Byzantine nodes. The vertex expansion of a graph $G = (V, E)$ on $n$ nodes is defined as
\begin{equation*}
    h(G)   =   \min_{0   <   |S|   \leq   \frac{n}{2}} \frac{|Out(S)|}{|S|}\text{,}
\end{equation*}
where $S$ is any subset of $V$ of size at most $\frac{n}{2}$ and $Out(S)$ is the set of neighbors of $S$ in $V \setminus S$. In particular, the seminal paper of Dwork et al.\ \cite{Dwork_1988}, which introduced and studied the problem of almost-everywhere Byzantine agreement in bounded degree graphs showed that such an agreement is achievable in \emph{almost all} $d$-regular graphs (i.e., all but a vanishingly small fraction of such graphs). We exploit the following fact in our current work: almost all $d$-regular graphs possess good expansion properties.

However, these algorithms assume knowledge of at least an estimate of the {\em size of the network} (in many cases, an estimate of the logarithm of the network size suffices) and related parameters such as the network diameter or the mixing time. In fact, the result of Dwork et al.\ assumes that all nodes know the global network topology. This suggests that it is non-trivial to design algorithms that work \emph{without} knowledge of these global network parameters in bounded-degree (or $d$-regular) expander networks. In such networks, nodes have a limited local view that is highly symmetric, and this enables Byzantine nodes to fake the presence (or absence) of parts of the network.

The goal of our algorithms is to guarantee that most of the honest (i.e., non-Byzantine) nodes obtain a good estimate of the network size. We note that obtaining ``almost-everywhere'' knowledge is the best one can hope for in such networks \cite{Dwork_1988}. Byzantine counting is related to, yet different from, other fundamental problems in distributed computing, namely, {\em Byzantine agreement} and {\em Byzantine leader election}. Similar to the latter two problems, it involves solving a global problem under the presence of Byzantine nodes. However, it is a different problem, since protocols for Byzantine agreement or leader election do not necessarily yield a protocol for Byzantine counting. In fact, many existing algorithms for these two problems (discussed below and in Section \ref{sec:results}) assume knowledge of $n$, the number of nodes in the network. In sparse networks, they require at least a reasonably good estimate of $n$, typically a constant factor estimate of $\log{n}$ is needed and usually sufficient (as explained in Section \ref{sec:results}). Indeed, one of the main motivations for this paper is to design distributed protocols in sparse networks that can work with little or no global knowledge, including the network size. An efficient protocol for the Byzantine counting problem can serve as a preprocessing step for protocols for Byzantine agreement, leader election, and other problems that either require or assume knowledge of an estimate of $\log n$ \cite{Augustine_2016} (cf.\ Section \ref{sec:results}).

Byzantine agreement and leader election have been studied extensively for several decades. Dwork et al.\ \cite{Dwork_1988}, Upfal \cite{Upfal_1994}, and King et al.\ \cite{King_2006_FOCS} studied the Byzantine agreement problem in \emph{sparse (bounded-degree) expander networks} under the condition of \emph{almost-everywhere} agreement, where {\em almost} all (honest) processors need to reach agreement as opposed to \emph{all} nodes agreeing as required in the standard Byzantine agreement problem. Dwork et.\ al.\ \cite{Dwork_1988} showed how one can achieve almost-everywhere agreement under up to $\Theta(\frac{n}{\log{n}})$  of Byzantine nodes in a bounded-degree \emph{expander} network ($n$ is the network size). Subsequently, Upfal~\cite{Upfal_1994} gave an improved protocol that can tolerate up to a linear number of faults in a bounded degree \emph{expander} of \emph{sufficiently large spectral gap}. These algorithms required polynomial number of rounds in the CONGEST model (where honest nodes send only small-sized messages) and   required $O(\log n)$ rounds in the LOCAL model (where there is no restriction on the message sizes) and polynomial (in $n$) number of messages. (For a comparison, similarly, our \textsf{Local} algorithm takes $O(\log{n})$ rounds and our \textsf{Congest} algorithm takes polynomial number of rounds.) Moreover, for Upfal's algorithm  the local computation required by each processor is exponential. The work of King et al.\cite{King_2006_FOCS} was the first to study scalable (polylogarithmic communication and number of rounds, and polylogarithmic computation per processor) algorithms for Byzantine leader election and agreement. All of the above algorithms require knowledge of the network topology (including the knowledge of $n$) --- nodes need to have this information hardcoded from the very start.

The works of \cite{Augustine_2012}, \cite{Augustine_2013_PODC}, and \cite{Augustine_2015_DISC} studied stable agreement, Byzantine agreement, and Byzantine leader election (respectively) in dynamic networks (see also \cite{Augustine_2016}), where in addition to Byzantine nodes there is also adversarial churn. All these works assume that there is an underlying bounded-degree regular expander graph (in fact, Dwork et al. among others assume $d$-regular random graphs which are expanders with high probability) and {\em all nodes are assumed to have knowledge of $n$}. It was not clear how to estimate $n$ without additional information under presence of Byzantine nodes in such (essentially, regular and constant degree expander) networks. In fact, the works of \cite{Augustine_2016, Augustine_2015_DISC} raised the question of designing protocols in expander networks that work when the network size is not known and may even change over time, with the goal of obtaining a protocol that works when nodes have strictly local knowledge. This requires devising a distributed protocol that can measure global network parameters such as size, diameter, average degree, etc.\ in the presence of Byzantine nodes in sparse networks, especially in sparse \emph{expander} networks.

Motivated by the above considerations, the work of Chatterjee et al.\ \cite{Chatterjee_2019} studied the Byzantine counting problem in a {\em ``small-world''} expander network under the assumption that the Byzantine nodes are {\em randomly} distributed (cf. Section \ref{sec:technical} for more details). They present a distributed algorithm running in polylogarithmic (in $n$) rounds in the CONGEST model that can output a constant factor estimate of $\log{n}$, where $n$ is the (unknown) network size under the presence of $O(n^{1 - \gamma})$  Byzantine nodes, where $\gamma > 0$ can be be any arbitrarily small (but fixed) constant. While this presents the first known Byzantine counting algorithm under this setting, it has two major drawbacks.

First, it does not work when Byzantine nodes are \emph{arbitrarily distributed} --- it crucially needs that they be \emph{randomly distributed}.

Second, it does not work for (just) expander networks; it needs additional structure, namely a {\em small-world} network, i.e., a network that has a large clustering coefficient.\footnote{i.e., a Watts-Strogatz type network similar to \cite{Watts_1998, Barthelemy_1999}.} The work of Chatterjee et al. crucially relies on the small-world property in its estimation of the network size. Hence the algorithm and techniques used in that paper \cite{Chatterjee_2019} are \emph{not} directly applicable to the present paper. Indeed, this paper uses a different approach compared to that of \cite{Chatterjee_2019} (cf. Section \ref{sec:technical}). While prior works on Byzantine agreement and leader election required only (sparse) expander networks \cite{Dwork_1988, Upfal_1994, King_2006_FOCS} under an arbitrary distribution of Byzantine nodes, Chatterjee et al.\ remark that:

\begin{quote}
    ``... for the Byzantine counting problem, which seems harder, however, expansion by itself does not seem to be sufficient.''
\end{quote}

In this paper, we show that Byzantine counting can indeed be solved in expander networks and almost all $d$-regular graphs under arbitrarily (adversarially) placed Byzantine nodes. This is the setting that is typically assumed in prior works on Byzantine agreement and leader election problems (e.g., \cite{Dwork_1988, Upfal_1994, King_2006_FOCS, Augustine_2012, Augustine_2013_PODC, Augustine_2015_DISC}).

Throughout this paper, we use the following terminology.
\begin{enumerate}
    \item We use the terms \emph{sparse network} and \emph{bounded-degree network} synonymously --- each describing a network where the maximum degree of a node is bounded by a constant, and hence the number of edges is linear in the number of vertices.
    
    \item A \emph{small-sized message} is defined to be one that contains $O(\log{n})$ bits in addition to at most a constant number of node IDs.
    
    \item We use the term \emph{most nodes} or \emph{most good nodes} to indicate $\geq  (1 - \beta)n$ nodes, where $n$ is the total number of nodes in the network (and is an unknown quantity in the context of this paper) and $\beta$ is any arbitrarily small (but fixed) positive constant.
    
    \item By \emph{efficient algorithms} we mean algorithms that use small-sized messages and run in $\polylog{(n)}$ time.
\end{enumerate}
\subsection{Our Contributions} \label{sec:results}

We present two distributed algorithms for the {\em Byzantine counting problem}, which is concerned with estimating the size (more specifically, the logarithm of the size, as considered here) of a sparse network in the presence of a large number of Byzantine nodes.

Let the network be denoted by $G = (V, E)$; let $n = |V|$ denote the (unknown) network size. Our first algorithm is {\em deterministic} and finishes in $O(\log{n})$ rounds in the \textsf{LOCAL} model and is time-optimal. This algorithm can tolerate up to $O(n^{1 - \gamma})$ adversarially placed Byzantine nodes for any arbitrarily small (but fixed) positive constant $\gamma$. It outputs a constant factor estimate of $\log{n}$ that is known to all but $o(1)$ fraction of the good nodes. This algorithm works for \emph{any} bounded degree expander network.

Our second algorithm is {\em randomized}. This algorithm works in \emph{almost all} $d$-regular graphs (i.e., all but a vanishingly small fraction of such graphs). We note that this is the same model used in the seminal work of Dwork et al.\cite{Dwork_1988}. Our algorithm works in the CONGEST model, where  honest nodes  use only \emph{small-sized} messages (unlike the first algorithm). See Section~\ref{sec:model} for more details about the network model. It tolerates up to $B(n) = n^{\frac{1}{2} - \xi}$ adversarially placed Byzantine nodes, where $\xi$ is any arbitrarily small (but) fixed positive constant. This algorithm takes $O(B(n)\log^2{n})$ rounds (hence $o(\sqrt{n})$ rounds for $B(n) = n^{\frac{1}{2} - \xi}$) and outputs a constant factor estimate of $\log{n}$ with  probability at least $1-o(1)$. The said estimate is known to at least $  (1 - \beta)n$ nodes for any arbitrarily small positive constant $\beta$.

We note that similar to our result, many prior Byzantine protocols \cite{Dwork_1988, Berman_1993_MST, Berman_1993_DC, Upfal_1994} in bounded-degree networks take a polynomial number of rounds in the \textsf{Congest} model (where honest nodes are limited to small-sized messages). However, all these protocols assume knowledge of $n$, and in certain cases, even the entire network topology. A notable exception is the protocol of King et al.\ \cite{King_2006_FOCS} that takes a polylogarithmic number of rounds in the \textsf{Congest} model, but this protocol also requires knowledge of the entire network topology (including the value of $n$). On the other hand, these protocols tolerate a substantially larger number of Byzantine nodes, even up to $\Theta(n)$ Byzantine nodes. It is unknown whether one can achieve the same level of fault-tolerance for Byzantine counting.

To complement our algorithms, we also present an impossibility result that shows that it is impossible to estimate the network size (or the logarithm of it) with any reasonable approximation and with any non-trivial probability of success if the network does not have sufficient vertex expansion. This shows that the assumption of the expansion property of the network is \emph{necessary} for solving Byzantine counting.

Both our algorithms are the first such algorithms that solve Byzantine counting in sparse, bounded degree networks under very general assumptions: they are fully local and need no global knowledge. Our algorithms can serve as a building block for implementing other non-trivial distributed computational tasks in Byzantine networks such as agreement and leader election where the network size (or its estimate) is not known a priori.

\paragraph{Applying our counting protocols.} To illustrate, we consider the {\em Byzantine agreement} protocol of \cite{Augustine_2013_PODC} that applies to sparse bounded-degree expander networks. It applies even when the network is dynamic with adversarial churn, but the network size is assumed to be stable. This protocol uses two main ideas to solve binary agreement, where the requirement is that most good nodes should decide on a common value ($0$ or $1$) which should be an input value of a good node: (1) {\em random walks} to sample nodes uniformly at random from the network and (2) {\em a majority protocol} to converge to the correct value. Both ideas require knowledge of  $\log{n}$, in particular, a constant factor upper bound of $\log{n}$. For random walks, $O(\log{n})$ is the {\em mixing time}, which is needed for walks to converge to the stationary distribution in a bounded degree expander; nodes need to know an upper bound on the mixing time to ensure that only sufficiently ``mixed'' random walks are used for sampling. The majority protocol uses the following simple idea: In one iteration, each node samples two random nodes and updates its value to the majority value among the three values: its own value and the two other values. It is shown that $O(\log{n})$ iterations are needed to converge to the almost-everywhere agreement with high probability, provided the number of Byzantine nodes is bounded by $O(\sqrt{n})$.

It is important to note that the above protocol assumes knowledge of $c\log{n}$, for some constant $c > 1$. However, using the Byzantine counting protocol of this paper as a preprocessing step, the above assumption can be removed. The counting protocol ensures that most honest nodes have a constant factor estimate of $\log{n}$ (this constant is fixed in the analysis). Although the counting protocol does not guarantee that all (or most) honest nodes have the \emph{same} estimate of $\log n$, it is easy to ensure that most honest nodes have an estimate that is some constant factor larger than $\log{n}$. This estimate suffices to run the Byzantine agreement protocol of \cite{Augustine_2013_PODC}.
\subsection{Technical Challenges and Drawbacks of Previous Approaches} \label{sec:technical}

The main challenge is designing and analyzing distributed algorithms in the presence of Byzantine nodes in networks where the honest nodes have only local knowledge, i.e., knowledge of their immediate neighborhood. For example, in a constant degree regular network, a node's local view does not yield any information on the network size. It is possible to solve the counting problem exactly in networks \emph{without} Byzantine nodes by simply building a spanning tree and converge-casting the nodes' counts to the root, which in turn can compute the total number of nodes in the network. A more robust and alternate way that works also in the case of \emph{anonymous} networks is the technique of {\em support estimation} \cite{Augustine_2012, Augustine_2016} which uses \emph{exponential} distribution. Alternatively, one can use a geometric distribution (see e.g., \cite{Kutten_2015, Pandurangan_2019_Book, Newport_2018}) to accurately estimate the network size.

Consider the following simple protocol for estimating the network size that uses the geometric distribution. Each node $u$ flips an unbiased coin until the outcome is heads; let $X_u$ denote the random variable that denotes the number of times that $u$ needs to flip its coin. Then, nodes exchange their respective values of $X_u$ whereas each node only forwards the highest value of $X_u$ (once) that it has seen so far. We observe that $X_u$ is geometrically distributed and denote its global maximum by $\bar{X}$; it can be shown that $\bar{X} = \Theta(\log{n})$ with high probability and hence can be used to estimate $\log{n}$.

The geometric distribution protocol fails when even just one Byzantine node is present. Byzantine nodes can fake the maximum value or can stop the correct maximum value from spreading and hence can violate any desired approximation guarantee. The work of \cite{Chatterjee_2019} successfully adapts the geometric distribution to work for their purpose. However, their work \cite{Chatterjee_2019} assumes additional structural properties of the network --- they assume ``small-world'' networks, i.e., networks with constant expansion \emph{and} large clustering coefficient. The latter property implies that for every node, many of its neighbors are well-connected among themselves. The protocol of \cite{Chatterjee_2019} exploits this fact to detect fake values sent by Byzantine nodes. This protocol does not work for graphs that {\em only} have the expander property (which as we show in the impossibility result is needed to estimate the network size within a non-trivial factor). Hence a new approach is needed as shown in this paper.

The work of \cite{Chatterjee_2019} also assumes that the Byzantine nodes are \emph{randomly} distributed in the network. This assumption coupled with the fact that their number is only $O(n^{1-\gamma})$ (where $\gamma$ is any arbitrarily small, but fixed, positive constant), results in (with high probability) every honest node having a significant number of honest neighbors (the number of neighbors depends on $\gamma$). The algorithm of \cite{Chatterjee_2019} \emph{fails} to work for expander networks with arbitrary or adversarial Byzantine node distribution, which is typically assumed in previous works on Byzantine protocols \cite{Dwork_1988, Upfal_1994, King_2006_FOCS, Augustine_2012, Augustine_2013_PODC, Augustine_2015_DISC}.

Prior localized techniques that have been used successfully for solving other problems such as Byzantine agreement and leader election such as random walks and majority agreement (e.g., \cite{Augustine_2013_PODC, Augustine_2015_DISC}) do not imply efficient algorithms for Byzantine counting. For instance, random walk-based techniques crucially exploit a uniform sampling of tokens (generated by nodes) after $\Theta(\text{mixing time})$ number of steps. However, the main difficulty in this approach is that the mixing time is unknown (since the network size is unknown) --- and hence it is unclear a priori how many random walk steps the tokens should take. Similar approaches based on the return time of random walks fail due to long random walks having a high chance of encountering a Byzantine node.

One can also use ``birthday paradox'' ideas to try to estimate $n$, e.g., as in the work of \cite{Ganesh_2007} in a non-Byzantine setting. However it fails too in the Byzantine case.

We note that one can possibly solve Byzantine counting if one can solve Byzantine leader election, as observed in \cite{Chatterjee_2019}, however, all known algorithms for Byzantine leader election (or  agreement) {\em assume a priori knowledge (or at least a good estimate) of the network size}.  Hence we require a new protocol that solves Byzantine counting from ``scratch''. In our network model, where most nodes, with high probability, see (essentially) the same local topological structure (and constant degree) even for a reasonably large neighborhood radius (see Lemma \ref{lemma-most-nodes-are-locally-tree-like-preliminaries}), it is difficult for nodes to break symmetry or gain a priori knowledge of $n$.

We point out that with constant probability, in our network model, due to the property of the $d$-regular random graph, an expected constant number of nodes might have multi-edges --- this can potentially be used to break ties; however, this approach \emph{fails} with constant probability.

Another approach is to try to estimate the diameter of the network, which, being $\Theta(\log{n})$ for sparse expanders, can be used to deduce an approximation of the network size. Assuming that there exists a leader in the network, one way to do this is for the leader to initiate the flooding of a message and it can be shown that a large fraction of nodes (say a $(1 - \epsilon)$-fraction, for some small $\epsilon > 0$) can estimate the diameter by recording the time when they see the first token, since we assume a synchronous network. However, this method fails since it is not clear, how to break symmetry initially by choosing a leader --- this by itself appears to be a hard problem in the Byzantine setting without knowledge of $n$ (or an estimate of $\log{n}$).
\subsection{A High-level Description of Our Protocols}

We now give a high-level intuition behind our protocols. Our first protocol works for \emph{any} expander network as long as the nodes have knowledge of some lower bound on the expansion. The main idea is to show that honest nodes that have a sufficiently large distance from any of the Byzantine nodes will be able to detect any deviations in the network structure caused by Byzantine nodes. The honest nodes can accomplish that by checking the expansion of their $i$-hop neighborhood, for some $i = \Omega(\log{n})$. This algorithm is \emph{time-optimal} and runs in time proportional to the network diameter. However, it is designed for the \textsf{Local} model, as the expansion check requires nodes to send messages of polynomial size.

The second algorithm achieves Byzantine counting by ensuring that most good nodes will send only small-sized messages. The main idea here is the following. The algorithm proceeds in phases. In phase $i$, $i$ is the current estimate of $\log(n)$. In each $i$-hop neighborhood of some node consisting only of good nodes, there are likely to be $\Theta(i)$ nodes that are generating \emph{beacon messages}, which are propagated for at least $i$ rounds through the network.

Upon receiving a beacon message, a node assumes that the value of $i$ is not yet too large and hence proceeds without \emph{deciding}. On the other hand, the probability of any good node generating a beacon message becomes $\frac{1}{\poly{(n)}}$ once $i = \Omega(\log{n})$, and hence good nodes that do not observe a beacon message within $O(i)$ rounds of phase $i$, decide on $i$ as their estimate.

To avoid the scenario where Byzantine nodes simply keep generating new beacon messages (to falsely induce a larger network size), the algorithm implements a \emph{blacklisting mechanism} that uses properties of random regular graphs to prevent nodes from generating multiple beacon messages within the same phase. This ensures that the Byzantine nodes will be blacklisted if they attempt to generate fake beacon messages.
\subsection{Other Related Work} \label{sec:related}

There have been several works on estimating the size of the network, see e.g., the works of \cite{Ganesh_2007, Horowitz_2003, Luna_2014, Terelius_2012, Shafaat_2008}, but all these works do not work under the presence of Byzantine adversaries. There have been some work on using network coding for designing byzantine protocols (see e.g., \cite{Jaggi_2008}); but these protocols have polynomial message sizes and are highly inefficient for problems such as counting, where the output size is small. There are also some works on topology discovery problems under Byzantine setting (e.g., \cite{Nesterenko_2006}), but these do not solve the counting problem.

Several recent works deal with Byzantine agreement, Byzantine leader election, and fault-tolerant protocols in dynamic networks.  We refer to \cite{Guerraoui_2013, Augustine_2012, Augustine_2013_PODC, Augustine_2013_SPAA, Augustine_2015_DISC} and the references therein for details on these works. These works crucially assume the knowledge of the network size (or at least an estimate of it) and don't work if the network size is not known.

There has been significant work in designing peer-to-peer networks that are provably robust to a large number of Byzantine faults \cite{Fiat_2002, Hildrum_2003, Naor_2003, Scheideler_2005}. These focus only on (robustly) enabling storing and retrieving data items. The works of \cite{King_2006_FOCS, King_2014, King_2006_SODA} address the Byzantine agreement problem, and the work of  \cite{Guerraoui_2013} presents a solution for maintaining a clustering of the network. In particular, \cite{King_2014} use a spectral technique to ``blacklist'' malicious nodes leading to faster and more efficient Byzantine agreement. All these works assume a sufficiently good estimate of the network size; in particular, none of them solves the Byzantine counting problem in sparse networks.

The work of \cite{Bortnikov_2009} shows how to implement uniform sampling in a peer-to-peer system under the presence of Byzantine nodes where each node maintains a local ``view'' of the active nodes. We point out that the choice of the view size and the sample list size of $\Theta(n^{\frac{1}{3}})$ necessary for withstanding adversarial attacks requires the nodes to have a priori knowledge of a polynomial estimate of the network size. \cite{Horowitz_2003} considers a dynamically changing network \emph{without} Byzantine nodes where nodes can join and leave over time and provides a local distributed protocol that achieves a polynomial estimate of the network size.

In \cite{Bovenkamp_2012}, the authors present a gossip-based algorithm for computing aggregate values in large dynamic networks (but without the presence of Byzantine failures), which can be used to obtain an estimate of the network size. The work of \cite{Chlebus_2009} focuses on the consensus problem under crash failures and assumes knowledge of $\log{n}$, where $n$ is the network size. Lenzen et al.\ \cite{Lenzen_2017} study the synchronous counting problem under Byzantine nodes which is a different problem: the goal here is to synchronize pulses among correct nodes. They study the problem in a complete network, and hence the network size is trivially known.

\paragraph{Byzantine fault detection in the context of asynchronous distributed systems.} There have also been several works on Byzantine fault detection --- see, e.g., \cite{Alvisi_2001}, \cite{Kihlstrom_2003}, \cite{Haeberlen_2006}, and \cite{Greve_2012}. Alvisi et al.\ \cite{Alvisi_2001} consider the problem of fault detection in Byzantine \emph{quorum systems} and design statistical methods to compute the current number of failures at any point of time. Their model assumes the knowledge of $n$, the total number of servers, whereas their goal is also clearly different. There have also been works on Byzantine fault detectors \cite{Kihlstrom_2003, Haeberlen_2006}, but these assume the complete graph where the knowledge of $n$ becomes trivial.

Kihlstrom et al.\ \cite{Kihlstrom_2003} propose and analyze new classes of Byzantine fault detectors to solve the consensus problem in an asynchronous distributed system of $n$ processes, in which the number of (Byzantine-) faulty processors is strictly less than $\frac{n}{3}$. Haeberlen et al.\ \cite{Haeberlen_2006} proposes a new idea for Byzantine fault detection by achieving \emph{eventual strong completeness} where every faulty node is eventually blacklisted by every correct node. The underlying communication graph is a \emph{complete graph} in both the network models of \cite{Kihlstrom_2003} and \cite{Haeberlen_2006}, thus the knowledge of $n$ becomes immediate and trivial.

Greve et al.\ \cite{Greve_2012} design and analyze a powerful Byzantine failure detector that works in dynamic distributed systems, where both the number of processors and the topology of the communication graph can change from round to round. Their work does not assume any knowledge of $n$; however, their work does \emph{not} solve the Byzantine counting problem either --- \emph{no} estimate about the \emph{global} network size can be made during the execution of their algorithm.

\section{Computing Model and Problem Definition} \label{sec:model}

\paragraph{The distributed computing model.} We consider a synchronous network represented by a graph $G$ whose nodes execute a distributed algorithm and whose edges represent connectivity in the network. The computation proceeds in synchronous rounds, i.e., we assume that nodes run at the same processing speed (and have access to a synchronized clock) and any message that is sent by some node $u$ to its  neighbors in some round $r\ge 1$ will be received by the end of round $r$. We consider the \textsf{Local} model, where there is no restriction on the size of the messages that can be transmitted per edge per round, \cite{Pandurangan_2019_Book, Peleg_2000_Book}; but we point out that our second algorithm ensures that most good nodes send only small-sized messages.

As is usual \cite{Pandurangan_2019_Book, Peleg_2000_Book}, we assume local computation (within a node) is free and instantaneous.

\paragraph{Byzantine nodes.} Among the $n$ nodes ($n$ or its estimate is not known to the nodes initially), up to $B(n)$ can be \emph{Byzantine}. The Byzantine nodes have unbounded computational power and can deviate arbitrarily from the protocol. This setting is commonly referred to as the \emph{full information model}.

We say that a node $u$ is \emph{good} or \emph{honest} if $u$ is not a Byzantine node. Byzantine nodes are \emph{adaptive} --- they have complete knowledge of the entire states of all nodes at the beginning of every round (including random choices made by all the nodes), and thus can take the current state of the computation into account when determining their next action. The Byzantine nodes also know the future random choices of the honest nodes, i.e., the Byzantine nodes are \emph{omniscient}. We assume that the Byzantine nodes are {\em arbitrarily} distributed in the network and that when a Byzantine node sends a message over an edge, it cannot fake its id. We note that both of these assumptions are quite typical in the literature \cite{Dwork_1988, Upfal_1994, King_2006_FOCS, Augustine_2012, Augustine_2013_PODC, Augustine_2015_DISC}.

\paragraph{Distinct IDs.} We assume that all nodes (including the Byzantine nodes) have {\em distinct IDs}, chosen from an arbitrarily large set whose size is unknown a priori. In other words, the node IDs can be viewed as comparable black boxes that do not leak any information about the network size. We point out that this precludes most nodes from estimating $\log{n}$ by looking at the length of their IDs.

\paragraph{Network topology for the first (deterministic) algorithm.} Let $G = (V, E)$ be the graph representing the network. We assume $G$ to be a bounded-degree expander network. For the sake of a self-contained exposition, we recall the definition of \emph{vertex expansion} below.

\begin{definition}[Vertex expansion of a graph]
    The \emph{vertex expansion} of a graph $G = (V, E)$ on $n$ nodes is defined as
    \begin{equation*}
        h(G)   =   \min_{0   <   |S|   \leq   \frac{n}{2}} \frac{|Out(S)|}{|S|}\text{,}
    \end{equation*}
    where $S$ is any subset of $V$ of size at most $\frac{n}{2}$ and $Out(S)$ is the set of neighbors of $S$ in $V \setminus S$.
\end{definition}

We assume that the network graph $G$ has a constant \emph{vertex expansion} $\alpha > 0$, where $\alpha$ is a fixed positive constant.

\paragraph{Network topology for the second (randomized) algorithm.} Here we assume $G$ to be a random $d$-regular graph model ($d$ is a constant) that is constructed by the union of $\frac{d}{2}$ (assume $d \geq 8$ is an even constant) random Hamiltonian cycles of $n$ nodes. We call this random graph model the \emph{$H(n, d)$ random graph model}, also called the \emph{permutation model} (please refer to \cite{Chatterjee_2021_arXiv} for a detailed exposition of the $H(n, d)$ random graph model and its various properties). It is known that such a random graph is an expander (in fact a Ramanujan expander \cite{Friedman_1991, Law_2003}) with high probability. The $H(n, d)$  model is a well-studied and popular random graph model (see e.g., \cite{Wormald_1999}), and has been used as a model for peer-to-peer networks and self-healing networks \cite{Law_2003, Pandurangan_2014}.

We note that the usual $d$-regular random graph model is the model where a graph is selected with uniform probability among all (simple) $d$-regular graphs \cite{Wormald_1999}. Thus if one can show a result that holds with high probability in a $d$-regular random graph, then it holds for \emph{almost all} $d$-regular graphs (as in Dwork et al\cite{Dwork_1988}). Since it is hard to work directly with the above model, one usually works with the so-called \emph{configuration (or pairing) model} \cite{Bollobas_1981} that can be used to generate a $d$-regular random graph. The advantage of the configuration model is that if one can show a high probability bound on the configuration model, then this implies a similar bound for $d$-regular ($d$ is a constant) random graphs \cite{Dwork_1988}. The configuration model is closely related to the $H(n, d)$ (i.e., permutation) model, which is sometimes easier to work with compared to the configuration model. It was shown by Greenhill et al.\ \cite{Greenhill_2002} that an event that holds with probability at least $1 - o(1)$ in the configuration model also holds with probability $1 - o(1)$ in the $H(n, d)$ model and vice versa. Thus, the results that we show for the $H(n, d)$ model also hold for the configuration model with probability at least $1 - o(1)$. Therefore they also hold for $d$-regular random graphs with the same probability. Hence they hold for \emph{almost all} $d$-regular graphs.

\paragraph{Problem definition.} Since we assume a \emph{sparse} (constant bounded degree) network and a large number of Byzantine nodes, it is difficult to ensure that every honest node eventually knows an exact estimate of $n$. This motivates us to consider the following ``approximate, almost everywhere'' variant of counting:

\begin{definition}[Byzantine counting] \label{definition-problem-definition}
    Suppose that there are $B(n)$ Byzantine nodes in the network and let $\epsilon$ be an arbitrarily small (but fixed) positive constant. We say that an algorithm solves Byzantine Counting in $T$ rounds if the following properties hold in all runs:
    \begin{enumerate}
        \item Every honest node $u$ (irrevocably) decides on an estimate of $\log{n}$, denoted by $\mathcal{L}_u$, within $T$ rounds.
		\item There is a set $S$ of at least $(1 - \epsilon)n - B(n)$ honest nodes such that each $u \in S$ has a constant factor estimate of $\log{n}$; i.e., there are fixed constants $c_1, c_2 > 0$, such that
            \begin{equation*}
                c_1\log{n}  \leq  \mathcal{L}_u  \leq  c_2\log{n}\text{.}
            \end{equation*}
    \end{enumerate}
\end{definition}

\paragraph{Some terminology.} We recall the following terminology that are used throughout this paper.
\begin{enumerate}
    \item We use the terms \emph{sparse network} and \emph{bounded-degree network} synonymously --- each describing a network where the maximum degree of a node is bounded by a constant, and hence the number of edges is linear in the number of vertices.
    
    \item A \emph{small-sized message} is defined to be one that contains $O(\log{n})$ bits in addition to at most a constant number of node IDs.
    
    \item We use the term \emph{most nodes} or \emph{most good nodes} to indicate $\geq  (1 - \beta)n$ nodes, where $n$ is the total number of nodes in the network (and is an unknown quantity in the context of this paper) and $\beta$ is any arbitrarily small (but fixed) positive constant.
    
    \item By \emph{efficient algorithms} we mean algorithms that use small-sized messages and run in $\polylog{(n)}$ time.
\end{enumerate}


\section{Preliminaries}

We use the notation $\ball_G(u,i)$ to refer to the inclusive $i$-hop neighborhood of node $u$ in graph $G$ and we omit $G$ when it is clear from the context. For a set of nodes $S$, we define $\ball_G(S,i) = \bigcup_{u \in S}\ball_G(u,i)$.

Both of our algorithms make use of a structural result that shows that Byzantine nodes have a somewhat limited impact on \emph{most} good nodes in expander graphs.

\begin{lemma} \label{lem:goodNodes}
    Consider an $n$-node graph $G = (V, E)$ with maximum degree $\Delta = O(1)$ and vertex expansion $\alpha>0$. Let $\byz$, an arbitrary subset of $V$, denote the set of Byzantine nodes with the restriction that $|\byz| \leq n^{1-\gamma}$, where $\gamma$ is any arbitrarily small (but fixed) positive constant. Then, for any $o(n)$-sized $F  \subset  V$, there exists a set $\good \subseteq V\setminus F$ of good nodes such that
      \begin{align}
          |\good| \ge n - 2|F| - o(n). \label{eq:goodsize}
      \end{align}
    Moreover, for each $u \in \good$, the following hold:
    \begin{enumerate}
        \item $\ball(u,\lfloor\tfrac{\gamma}{2}\log_\Delta n\rfloor)$ does not contain any Byzantine nodes.

        \item Let $H$ be the subgraph induced by nodes in $\good$. Then, for any constant $c>0$ such that $|\ball_H(u,c\log n)| \le \tfrac{|\good|}{2}$, it holds that every vertex subset $S \subseteq \ball_H(u,c\log n)$ has a vertex expansion of $\ge \alpha'$ in graph $H$, for any fixed constant $\alpha'  <  \alpha$.
    \end{enumerate}
\end{lemma}
\begin{proof}
    Consider the set $\byz$ of Byzantine nodes. We first instantiate Lemma \ref{lem:culling} by removing the set $V(G) \setminus (\byz \cup F)$ from $G$, and thus obtain a connected subgraph $H$ of size $\ge n - o(n)$. By Lemma \ref{lem:culling}, $H$ contains at least
    \begin{center}
        $n - (|F| + |\byz|)\lb( 1 + \frac{1}{\phi(1-c')}\rb) = n - |F| \lb( 1 + \frac{1}{\phi(1-c')} \rb) - o(n)$
    \end{center}
    good nodes. Also, every one of its subsets of size $\geq  \frac{|H|}{2}$ has a vertex expansion of at least $\alpha'$, for any constants $\alpha'  <  \alpha$ and $c'  <  1$. Choosing $c' = 1 - \tfrac{1}{\phi}$, implies that $H$ contains $n - 2|F| - o(n)$ nodes, as required. Assuming a maximum degree of $\Delta$, we get $|\ball(\byz,j)| \le |\byz|\Delta^j$, for any $j\ge 0$.

    We observe that
    \begin{center}
        $\big|\ball(\byz,\lfloor\tfrac{\gamma}{2}\log_\Delta n\rfloor)\big| \le |\byz| \cdot \Delta^{(\gamma/2)\log_\Delta n} = n^{1 - \gamma/2}$.
    \end{center}

    It follows that the set $\good = V(H) \setminus \ball(\byz,\tfrac{\gamma}{2}\log_\Delta n)$ satisfies (1) and (2).
\end{proof}
\subsection{The ``locally tree-like'' property of an $H(n,d)$ random graph} \label{sec:treelike}

We refer to \cite{Chatterjee_2021_arXiv} for a detailed exposition of the $H(n, d)$ random graph model and its various properties. For the sake of completeness, we merely state the main definitions and lemmas needed here. The ``locally tree-like'' property of an $H(n,d)$ random graph says that for most nodes $w$, the subgraph induced by $B(w,r)$ up to a certain radius $r$ looks like a tree. More specifically, let $G$ be an $H(n,d)$ random graph and $w$ be any node in $G$. Consider the subgraph induced by $B(w,r)$ for $r  =  \frac{\log{n}}{10\log{d}}$. Let $u$ be any node in $Bd(w,j)$, $1 \leq j < r$. $u$ is said to be \emph{typical} if $u$ has only one neighbor in $Bd(w,j-1)$ and $(d-1)$-neighbors in $Bd(w,j+1)$; otherwise it is called \emph{atypical}.

\begin{definition}[Locally Tree-Like Property] \label{defn-locally-tree-like-node-preliminaries}
    We call a node $w$ \emph{locally tree-like} if no node in $B(w,r)$ is atypical. In other words, $w$ is locally tree-like if the subgraph induced by $B(w,r)$ is a $(d-1)$-ary tree.
\end{definition}

Using the properties of the $H(n, d)$ random graph model and standard concentration bounds, it can be shown that most nodes in $G$ are locally tree-like:
\begin{lemma} \label{lemma-most-nodes-are-locally-tree-like-preliminaries}
    In an $H(n,d)$ random graph, with high probability, at least $n - O(n^{0.8})$ nodes are locally tree-like.
\end{lemma}

The proof of this lemma as well as a more detailed discussion of the $H(n, d)$ random graph model can be found in the appendix of \cite{Chatterjee_2021_arXiv}.


\section{A Time-Optimal Deterministic Algorithm} \label{sec:local}

In this section, we present and analyze a simple algorithm that solves the Byzantine counting problem in the \textsf{Local} model --- see Algorithm \ref{alg:local} for the pseudocode.

Our goal here is to show that a set called $\good$ consisting of $\geq  n - o(n)$ good nodes achieve a constant factor approximation of $\log{n}$ when executing our algorithm. Lemma \ref{lem:goodNodes} formalizes the criteria for a good node to be in $\good$: in particular, a good node needs to have a distance of $\Omega(\log{n})$ from all Byzantine nodes and the graph induced by $\good$ must have nearly the same vertex expansion as the original network.

\subsection{Description of the algorithm}

Throughout the algorithm, each node $u$ locally builds an approximation of its $i$-hop neighborhood for rounds $i  =  1, 2, 3, \ldots$, which we denote by $\hat{B}(u,i)$. To this end, we instruct nodes to simply forward the content of their current $\hat{B}(u,i)$ at the start of round $i$. Considering that we assume (at most) $n^{1-\gamma}$ Byzantine nodes, node $u$ needs to be careful when integrating any newly received knowledge.

There are two possibilities for triggering a decision of node $u$. Firstly, $u$ immediately decides if it notices some structural inconsistencies in the received topology information, such as a degree larger than $\Delta$, or the addition of spurious edges to vertices that it had already learned about previously (cf.\ Line~\ref{line:inconsistent}).

Furthermore, after obtaining $\hat{B}(u, i+1)$ by adding the received topology information in round $r$ into $\hat{B}(u,i)$, node $u$ also decides if any of the subsets of $\hat{B}(u,i)$ do not have sufficient vertex expansion with respect to $\hat{B}(u,i+1)$. Intuitively speaking, this second condition ensures that Byzantine nodes cannot trick $u$ into continuing forever. The algorithm's correctness crucially rests on the original network having constant expansion --- a point that is further emphasized by our impossibility result in Theorem~\ref{thm:impossibility}.

\begin{remark}
    We observe that, for $o(n)$ nodes in $G \setminus \good$, the adversary essentially controls the termination time. This is \emph{not} simply a drawback of our algorithm, but, instead, \emph{unavoidable} when assuming a worst-case placement of Byzantine nodes in the network: For instance, consider a $d$-regular expander and a set of $\lfloor n^{1-\gamma}/d\rfloor$ good nodes $U$ that are surrounded by roughly $n^{1-\gamma}$ Byzantine nodes, i.e., none of the edges emanating from $U$ to $G\setminus U$ are connected to good nodes. Then, the Byzantine nodes could simply send information corresponding to a large fake network of some arbitrary size $n'$ with sufficiently high expansion to the nodes in $U$. It is easy to see that no algorithm can distinguish this case from $U$ being indeed part of a network of size $n'$.
\end{remark}

We state below the main result of this section. The rest of this section is devoted to proving it.
\begin{theorem} \label{thm:local}
    Let $\gamma \in (0,1)$ and $\Delta > 0$ be arbitrary fixed constants. Consider an $n$-node network with a maximum node degree bounded by $\Delta$ and a constant vertex expansion $\alpha$. There exists a deterministic \textsf{LOCAL} algorithm such that $n - o(n)$ good nodes decide on a $\left(\frac{\gamma}{2\log\Delta}\right)$-approximation of $\log{n}$ in $O(\log{n})$ rounds in the presence of up to $n^{1-\gamma}$ arbitrarily (adversarially) placed Byzantine nodes.
\end{theorem}
\begin{algorithm}
\begin{algorithmic}[1]
    \State Let $\hat{B}(u,1)$ be the set of nodes in the inclusive neighborhood of $u$.
    
    \For{round $i=1,2,\dots$}
        \State Broadcast $\hat{B}(u,i)$ to all neighbors. 
        \State Let $I$ be the information received in the current round.
        \If{$\textsf{inconsistent}(\hat{B}(u,i),I)$ or (some neighbor is mute)}
            \State Decide on $i$ and terminate. \label{line:inconsistent}
        \EndIf
        \State Create $\hat{B}(u,i+1)$ by incorporating $I$ into $\hat{B}(u,i)$.
        \For{each vertex subset $S \subseteq V(\hat{B}(u,i))$}
            \State Let $\alpha'>0$ be the constant given by Lemma~\ref{lem:goodNodes}.
            \If{$S$ does not have vertex expansion $\ge \alpha'$ in graph $\hat{B}(u,i+1)$} \label{line:expansionCheck}
                \State Decide on $i$; terminate.
            \EndIf
        \EndFor
    \EndFor
    
    \item[]
    
    \State \textbf{predicate} $\textsf{inconsistent}(\hat{B}(u,i),I)$ returns \texttt{true} iff 
    \State \qquad $I$ contains a node with degree $>  \Delta$, or \label{line:degreeCheck}
    \State \qquad $I$ contains a set of incident edges for some node $v$, but already $v \in \hat{B}(u,j)$ for some $j\le i-1$.
\end{algorithmic}
\caption{An $O(\log n)$ time algorithm in the \textsf{Local} model. Code for node $u$.} \label{alg:brounds}
\label{alg:local}
\end{algorithm}
\subsection{Analysis of the algorithm}

\begin{lemma} \label{lem:atLeastLOCAL}
    All nodes in $\good$ decide on a value of at least $\lb\lfloor\tfrac{\gamma}{2}\log_\Delta n\rb\rfloor$.
\end{lemma}
\begin{proof}
We will proceed by induction over the number of rounds. Consider the graph $H$ given by Lemma~\ref{lem:goodNodes} and recall that $u \in V(H)$ by definition. Since $H$ has a vertex expansion $\ge \alpha'$, it follows that $u$'s neighborhood (in $H$) has size at least $1 + \alpha'$, which guarantees that $u$ passes the expansion-check in Line \ref{line:expansionCheck} in Round $1$. Moreover, $u$ has a distance of at least $\lfloor\tfrac{\gamma}{2}\log_\Delta n\rfloor$ from any Byzantine node, and hence it does not receive any inconsistent information during the first $\lfloor\tfrac{\gamma}{2}\log_\Delta n\rfloor$ rounds. This ensures $u$ will not decide in Line~\ref{line:inconsistent} Round $1$, which completes the inductive base.

Now, consider the inductive step $1  <  i  <  \lfloor \tfrac{\gamma}{2}\log_\Delta n\rfloor$, and suppose that $u$ has not decided at the end of round $i-1$. Similarly as in the case $i = 1$, it holds that $u$ does not decide due to receiving inconsistent information. Moreover, note that
\begin{equation*}
    |\ball(u,i)| \le \Delta^{i+1} \le n^{\gamma/2} < \tfrac{|H|}{2}\text{,}
\end{equation*}
since $|H|\ge n - o(n)$ by Lemma~\ref{lem:goodNodes}. Hence every subset of $\hat{B}(u,i)$ is guaranteed to have a vertex expansion of at least $\alpha'$, which ensures that $u$ continues to round $i+1$ without deciding.
\end{proof}

The next lemma tells us that, if a good node $u$ that has not yet decided, then its local $i$-neighborhood approximation $\hat{B}(u,i)$ does not contain inconsistent information concerning the nodes in $\good$. We will make use of this property in Lemma~\ref{lem:atMostLOCAL} below.

\begin{lemma} \label{lem:inconsistent}
    Suppose that $u \in \good$ has not decided by the end of round $i$, and consider graph $H$ given by Lemma \ref{lem:goodNodes}. Then, for each $v \in \ball_H(u,i)$ and any node $w$, it holds that $e=\{v,w\} \in \hat{B}(u,i)$ if and only if $e \in E(G)$.
\end{lemma}
\begin{proof}
    By the definition of $H$, it follows that, for each $v \in \ball_H(u,i)$, there exists a shortest path $p  =  (v = p_1, p_2, \ldots, p_j = u)$ consisting of $j \le i$ good nodes connecting $u$ and $v$. Moreover, it is easy to see that node $p_k$ ($1\le k < j$) on path $p$ must have broadcast its topology information in round $i-j+k$, since otherwise its neighbor $p_{k+1}$ would have terminated at the end of round $i-j+k$, because of having noticed that $p_k$ was mute. This in turn would have cause $p_{k+1}$ to terminate at the end of round $i-j+k+1$ and hence would have propagated to $u$ by round $i$, contradicting the premise of the lemma.

    Since each of the good nodes in $p$ forwards the received topology information towards $u$, it follows that node $u$ receives a message from some good neighbor, which contains the exact set $E_v$ of edges incident to $v$ in $G$. Suppose that, in some round during round $i$, node $u$ also receives a message containing a set of edges $E_v' \ne E_v$ of $v$, possibly injected by Byzantine nodes. However, Line~\ref{line:inconsistent} ensures that $u$ decides instantly in round $i$, since it has added inconsistent information to $\hat{B}(u,i)$. This results in a contradiction.
\end{proof}

\begin{lemma} \label{lem:atMostLOCAL}
    Every node in $\good$ decides on a value of at most $\diam(G) + 1$.
\end{lemma}
\begin{proof}
    Assume toward a contradiction that there is a node $u \in \good$ that decides on a value strictly greater than $\diam(G) + 1$. By the description, of the algorithm, this means that $u$ did not decide when executing round $i$, where $i = \diam(G) + 1$. Consider the content of $\hat{B}(u,i)$ after receiving all messages for round $i$. Note that it is possible that $\hat{B}(u,i)$ also contains information that was injected by Byzantine nodes.

    Let $F$ denote the \emph{Byzantine part} of $\hat{B}(u,i)$, i.e.,
    \begin{center}
        $F  \defeq  \hat{B}(u,i) \setminus \good$
    \end{center}
    and call
    \begin{center}
        $R  \defeq  \hat{B}(u,i) \cap \good$
    \end{center}
    its \emph{honest part}.

    We can assume that Byzantine nodes do not send any inconsistent information regarding the graph induced by $R$, as otherwise $u$ will decide in round $i$ and we are done. Similarly, we can rule out that any node in $R$ has already decided: For if some $w$ decided and remained mute, this would cause its good neighbors to decide in the next round (cf.\ Line \ref{line:inconsistent}), which in turn would propagate (through good nodes) to $u$, causing it to decide. Consequently, Lemma~\ref{lem:inconsistent} tells us that all edges emanating from nodes in $R \setminus \byz$ in graph $\hat{B}(u,i)$ also exist in $G$. In particular, there are no edges between $R \setminus \byz$ and $F$. Since every node in $G$ has distance at most $\diam(G)  =  i - 1$ to $u$, it follows that $R \subseteq \hat{B}(u,i-1)$ and thus $u$ will check $R$'s vertex expansion with respect to graph $\hat{B}(u,i)$ at the end of round $i+1$.

    To complete the proof, we analyze the expansion-check in Line \ref{line:expansionCheck} for the set $R$. Observe that $R$ contains all nodes within distance $\diam(G)$ from $u$ in graph $H$ (see Lemma~\ref{lem:goodNodes}). Given that $\diam(H) \le \diam(G)$ and the fact that nodes in $\good$ are connected in $H$, we know that
    \begin{center}
        $|R| \ge |\good| = n - o(n)$.
    \end{center}

    Recall that there are at most $n^{1-\gamma}$ Byzantine nodes in $R \cap \hat{B}(u,i)$. Since we assumed that Byzantine nodes did not send inconsistent information, each Byzantine node has at most $\Delta$ neighbors in $F$ (cf.\ Line \ref{line:degreeCheck}). It follows that there is a set $S'$ of at most $\Delta n^{1-\gamma} = o(n)$ fake vertices in the set $(\hat{B}(u,i) \setminus R) \subseteq F$ that have an edge to $R$. To satisfy the expansion-check (see Line~\ref{line:expansionCheck}), the number of neighbors of vertices in $R$ would need to be $|R|(1+\alpha') = \Omega(n)$, far exceeding the $o(n)$ fake vertices in $S'$. Hence the expansion-check fails for set $R$, causing $u$ to decide on $i  =  \diam(G) + 1$.
\end{proof}

Combining Lemmas \ref{lem:atLeastLOCAL} and \ref{lem:atMostLOCAL} shows the claimed bound on the approximation achieved by the $n - o(n)$ nodes in set $\good$. From Lemma \ref{lem:atMostLOCAL}, it follows that the round complexity until all but $o(n)$ nodes have decided is $O(\diam(G)) = O(\log{n})$. This completes the proof of Theorem \ref{thm:local}.


\section{Byzantine Counting with Small Messages} \label{sec:congest}

We now describe an algorithm that guarantees most good nodes will achieve a constant factor approximation of $\log(n)$ while sending only messages of small size (proportional to the number of bits of any node's ID). We give the detailed correctness proof in Section \ref{sec:congestAnalysis}. As mentioned in Section \ref{sec:model}, our algorithm works in the $H(n,d)$ $d$-regular random graph model with high probability, i.e., with probability at least $1 - n^{-c}$, for some constant $c \geq 1$. As discussed in Section \ref{sec:model}, this implies that the algorithm works in almost all $d$-regular graphs with probability at least $1 - o(1)$.

In Algorithm \ref{alg:congest}, each node keeps track of its current estimate in a variable $i$ that is initialized to a fixed constant. A node increases $i$ whenever it enters a new \emph{phase}, where the goal of a phase is to determine whether $i$ is already a sufficiently good approximation of $\log(n)$. On the other hand, once a node concludes that its current value of $i$ is sufficiently large, it \emph{decides on $i$} and stops participating in future phases. Each phase $i$ consists of roughly $e^{(1-\gamma)i}+1$ \emph{iterations}, and each iteration of phase $i$ takes $2i+5$ rounds: During the first $i+2$ rounds, nodes disseminate so called ``beacon messages'' (described next) whereas, during the following $(i+3)$-rounds, all yet-undecided nodes ensure that everyone in their $(i+3)$-neighborhood knows that they have not yet decided by sending a ``continue'' message.

\paragraph{Beacon Messages and Path Fields.} At the start of an iteration, a good node $u$ chooses to become \emph{active} with probability $\Theta\!\lb(\frac{i}{d^i}\rb)$, where $d$ is $u$'s degree. The intuition behind this probability is that this ensures that on the average there are approximately $O(i)$ nodes that are active in a ball of radius $i$ --- note that the tree-like property of expander graphs (see Section \ref{sec:treelike}) ensures that the number of nodes in a ball is $\Theta(d^i)$.

If $u$ becomes active, it broadcasts a \emph{beacon message} to its neighbors, which is then forwarded for $i+2$ rounds. Intuitively speaking, these beacon messages signal to other nodes that they should not yet decide on their current estimate.

In more detail, a beacon message $\langle \texttt{beacon}, u, P \rangle$ has an \emph{origin id} $u$, and a \emph{path field} $P$, which is the path of nodes that the message has visited so far. That is, whenever the message is sent from a node $w$ to a node $v$, we append $v$ to the path field before forwarding the message. Of course, it is entirely possible that these fields contain bogus information if the message passed through a Byzantine node.

\paragraph{Blacklisting.} Whenever a node $v$ receives a beacon message, it inspects the attached path field $P  =  (u_1, u_2, \ldots, u_k)$ by performing a series of checks.

First, $v$ checks whether the neighbor from which it received the message does indeed have id $u_k$. If $v$ finds that the sender has an id different from $u_k$, it simply discards that message. Node $v$ also maintains a \emph{blacklist} set $BL$, which is reset at the start of each phase and is gradually filled throughout a phase's iterations.

In more detail, let us suppose that the above mentioned beacon message was the \emph{first} one received by $v$ in iteration $1$ of phase $i$, from some neighbor $w$. Then, $v$ adds all nodes except the ones in the {$\lfloor (1-\epsilon)\rfloor i$-suffix of $P$} to its blacklist $BL$, where this {suffix} consists of the last $\lfloor (1-\epsilon)\rfloor i$ nodes on the path to the destination $v$. The intuition behind this rule is that $u$ blindly trusts all nodes that are close to it, but won't accept another beacon message in the future if it has traversed the same (far away) nodes twice in this iteration.

In addition, $v$ sets its variable $\shortestPath \gets P$, which indicates the (supposedly) shortest path over which $v$ received a beacon message in this iteration. (If $v$ receives two or more beacon messages simultaneously, it discards all but one.) Note that $v$ resets $\shortestPath$ at the end of each iteration. Then, assuming that we are still in the first $i+1$ rounds of the iteration upon the reception of this beacon message, $v$ broadcasts the message $\langle \texttt{beacon}, P' \rangle$ with the modified path field $P'$ to its neighbors where $P'$ is obtained by appending $v$ to $P$.

As mentioned, blacklisting ensures that $v$ does not accept a beacon message if the message took a path leading through the same nodes from which it has already seen a beacon message in this phase. Blacklisting is implemented as follows: If $v$ receives a beacon message $m'$ in some iteration $\ell>1$ of phase $i$ and the node IDs contained in the path field that are at least $\lfloor(1-\epsilon)i\rfloor$ away from $v$ intersect with the nodes already added to $BL$ during the previous iterations, then $v$ will not use $m'$ to update its $\shortestPath$ variable. However, it is important to keep in mind that, even in this case, $v$ still broadcasts the message with the updated path field to its neighbors, assuming we are still in the first $(i+1)$-rounds of the iteration.

Consequently, if $(i+2)$ rounds have passed and node $v$ did not set $\shortestPath$ in this iteration either because it did not receive any beacon messages or all received beacons carried already blacklisted node IDs, then $v$ decides on its current value of $i$.

Introducing blacklisting avoids the scenario where Byzantine nodes keep generating new beacon messages that trigger good nodes to continue progressing to the next phase, possibly significantly overshooting the actual value of $\log(n)$ before deciding. The blacklisting mechanism kicks in once $i = \Omega(\log{n})$ since the algorithm ensures that (see Lemma \ref{lem:upperEnd}):
\begin{enumerate}
    \item there is no iteration in which a good node still generates a new beacon message (whp);
    \item the number of iterations performed in phase $i$ exceeds the number of Byzantine nodes.
\end{enumerate}

For instance, suppose that a Byzantine node $b$ generates a beacon message with a fake path field in iteration $1$. Even though $b$ can trick all good nodes into accepting this beacon message in this iteration, it will fail to convince a set $U$ of good nodes that have a distance of at least $\lfloor (1-\epsilon)i\rfloor$ from $b$ into accepting such a message in any future iteration of this phase.

To see why this is the case, observe that when a node $u \in U$ receives a message where $b$ was involved in faking the path field, $b$ will be added to $u$'s blacklist because its ID will not be in the $\lfloor (1 - \epsilon)i\rfloor$ suffix of the path field by the time the message reaches $u$ (cf.\ Line \ref{line:blacklist}), assuming that the message did not pass through other Byzantine nodes that are closer to $u$. (Recall that $i$ is large enough such that good nodes have ceased from generating beacon messages and hence every beacon message that is still in transit must have been injected by Byzantine nodes.) The upshot is that a good node $u$ that has all the Byzantine nodes at a distance of at least $\lfloor (1-\epsilon)i\rfloor$ will blacklist at least $1$ Byzantine node $b$ in each iteration if $b$ generates a beacon message. Hence, $u$ will encounter an iteration in which its $\shortestPath$ variable is not set, thus causing it to decide on $i$.

\paragraph{Technical challenges.} There are several technical difficulties that we need to handle in our correctness proof. For instance, we need to choose the probability of generating beacon messages in a way such that $i$ does not become too large before most nodes have reached a decision, as we might end up with a value of $i$ where almost all good nodes are \emph{within} distance $\lfloor (1 - \epsilon)i\rfloor$ of some Byzantine node, thus disarming the blacklisting mechanism.

On the other hand, the blacklisting process itself reduces the number of nodes that a good node considers for beacon messages, which may cause too many nodes to decide early due to not seeing a beacon message in each iteration. We use two techniques to avoid this second problem:
\begin{enumerate}
    \item We use the \emph{tree-like} property of the regular expander graphs (see Section \ref{sec:treelike}). This shows that the remaining nodes provide sufficient expansion even if a large number of paths have been invalidated due to at least one of their nodes being blacklisted.

    \item We instruct undecided nodes to send out \texttt{continue} messages that are forwarded for $(i + 3)$ rounds in phase $i$. Upon reception of such a message, a node that has possibly already decided and stopped increasing its phase counter, will become active again and generate beacon messages with the appropriate probability.
\end{enumerate}

We state below the main result of this section. The rest of this section is devoted to proving it.
\begin{theorem} \label{thm:congest}
    Let $\xi$ and $\beta$ be any arbitrarily small (but fixed) positive constants. Let $B(n)  =  n^{\frac{1}{2} - \xi}$ denote the number of Byzantine nodes in the network. Consider the $H(n,d)$ random regular graph $G$ of $n$ nodes with constant vertex expansion, where $d$ is a sufficiently large constant. Then there exists an algorithm such that, with high probability, at least $(1 - \beta)n$ nodes send messages of at most $O(\log{n})$ bits and decide on a constant factor approximation of $\log{n}$ in time $O(B(n) \cdot \log^2{n})$ in the presence of up to $B(n)$ arbitrarily (adversarially) placed Byzantine nodes.
\end{theorem}
\begin{algorithm}
\begin{small}
\begin{algorithmic}[1]
    \For{phase $i=c, c+1, \dots$} \Comment \textsl{$c \ge \tfrac{2\log 2}{(2-\delta)\eta}$ is a sufficiently large constant, see \eqref{eq:gamma}} \label{line:startPhase}
        \State Each node initializes its phase $i$ blacklist $BL = \emptyset$.
        \For{iteration $j=1,\dots,\lfloor e^{(1-\gamma)i}\rfloor +1$ of phase $i$} \Comment \textsl{each iteration takes $(2i + 5)$ rounds.} \label{line:startIteration}
            \State Each node $u$ initializes $\shortestPath \gets \texttt{none}$. 
            \State $u$ becomes active with probability $\frac{c_1\cdot i}{d^i}$, for a sufficiently large constant $c_1$. \label{pseudocode-line-activation-probability-and-also-the-definition-of-c1}
            
            \If{$u$ is active} 
                \State $u$ updates $\shortestPath \gets (u)$.
                \State $u$ broadcasts message $m=\langle \texttt{beacon},u,P \rangle$, $u$'s own id denotes the origin, 
                \State and $P$ is the path of nodes visited previously by $m$; i.e., $P  =  \texttt{none}$.
                \State Message $m$ is forwarded by correct nodes for exactly $(i + 2)$ rounds (see below).
            \EndIf
            
            \medskip
            
            \State \Comment \textbf{During the first $i+2$ rounds of the iteration:}
            
            \If{node $u$ receives a set of beacon messages from its neighbors} \label{line:sendMsgStart}
                \State $u$ discards all but one arbitrarily chosen message.
                \State Let $m=\langle \texttt{beacon},v,Q \rangle$ be this message; assume that it was received from neighbor $w$. \label{line:rcvd}
                \State $u$ appends $w$'s id to path $Q$ yielding $Q'$. 
                \If{we are still within the first $i$ rounds of the iteration} 
                    \State $u$ broadcasts $\langle \texttt{beacon},v,Q' \rangle$.  \label{line:sendMsgEnd}
                \EndIf
                \State Let $S$ be the set of all except the last $\lfloor \lb(1-\epsilon\rb) i\rfloor$ nodes in $Q'$, where $\epsilon$ is defined in \eqref{eq:eps}.\label{line:epsilonPrefix}
                \If{ $S \cap BL = \emptyset$} 
                    \If{$\shortestPath = \texttt{none}$}
                        \State $u$ updates $\shortestPath \gets Q'$. 
                    \EndIf
                \EndIf
            \EndIf
      
            \medskip
            
            \State \Comment \textbf{After the first $i+2$ rounds of iteration $j$:}
            \If{$\shortestPath=\texttt{none}$ at node $u$ and $u$ has not yet decided}
                \State $u$ decides on $i$.\label{line:decide}
            \EndIf
            
            \State Let $S$ be the set of all except the last $\lfloor(1-\epsilon) i\rfloor$ nodes in $\shortestPath$.
            \State $u$ adds all nodes in $S$ to blacklist $BL$. \label{line:blacklist}
        \EndFor \Comment\textsl{End of the $j^{\text{th}}$ iteration of phase $i$}

        \medskip

        \If{$u$ has not decided} 
            \State $u$ broadcasts a $\langle \texttt{continue} \rangle$ message, which is forwarded for $i+3$ rounds by other nodes. 
        \EndIf
        \State When receiving multiple $\langle \texttt{continue} \rangle$ messages simultaneously, all but one are discarded.
        \If{$v$ has decided and $v$ did not receive $\langle \texttt{continue} \rangle$ in $i+3$ rounds}
            \State $v$ \textbf{exits} the for-loop. 
        \EndIf
        \State \Comment\textsl{Iteration $j$ ends after the \texttt{continue} messages have been in transit for $i+3$ rounds.} \label{line:endIteration}
    \EndFor \Comment\textsl{End of phase $i$}

    \medskip

    \If{$v$ has decided (possibly in some earlier phase and $v$ receives a $\langle \texttt{continue} \rangle$ message}
        \State $v$ reenters the for-loop and updates its value of $i$ to the current phase value. (It can do so by keeping track of the number of rounds since starting.) \label{line:reenter}
    \EndIf
\end{algorithmic}
\end{small}
\caption{A Byzantine-Resilient Counting algorithm using messages of size $O(\log{n})$ (at most nodes), assuming $O(n^{1-\gamma})$ Byzantine nodes. Nodes do not have any other global knowledge apart from $\gamma$.} \label{alg:congest}
\end{algorithm}
\subsection{Analysis of Algorithm~2} \label{sec:congestAnalysis}

For the analysis, we assume at most $n^{1 - \gamma}$ Byzantine nodes, where $\gamma$ needs to satisfy
\begin{equation} \label{eq:gamma}
  \gamma   \geq   \frac{1}{2 - \delta}  +  \eta,
\end{equation}
for any fixed constants $0  <  \delta  \leq  \frac{1}{2}$ and $\eta > 0$. Note that the smaller $\delta$ is, the smaller $\gamma$ is. Therefore maximum Byzantine tolerance is achieved when $\delta$ is very close to (but slightly greater than) zero and $\gamma$ is very close to (but slightly greater than) $\frac{1}{2}$. In that case, the maximum Byzantine tolerance, i.e., the maximum number of Byzantine nodes that our algorithm can tolerate, boils down to $n^{\frac{1}{2} - \xi}$, as stated in Theorem~\ref{thm:congest}.

The parameter $\epsilon$ that we use to determine the distance outside of which the blacklisting becomes effective in our algorithm, is fixed as
\begin{equation} \label{eq:eps}
    \epsilon =  1 - \frac{(1-\delta)}{\log d}\gamma\text{.}
\end{equation}

Let $\goodtree  =  \good \cap \treelike$ be the set of nodes that have a sufficiently large distance to all Byzantine nodes due to being in set $\good$, and that also have the property of $d$-ary trees up to some radius of length $\frac{\log_d{n}}{10}$ (see Section \ref{sec:treelike}).

We will first study the progress of the algorithm at nodes in $\goodtree$, for the phases up to radius $\rho$, where
\begin{equation} \label{eq:rho}
    \rho   =   \left\lfloor  \min\lb({(1-\delta)\gamma}\log_d{n}, \frac{1}{10}\log_d{n}\rb)\right\rfloor - 2,
\end{equation}
since, in phase $i$, we require the tree-like property to hold up until radius $(i + 2)$. We also recall that $c_1$ is any large constant, as defined in Line~\ref{pseudocode-line-activation-probability-and-also-the-definition-of-c1} of the pseudocode of Algorithm \ref{alg:congest}.
\subsubsection{Analysis of the early phases of the algorithm: when $i < \rho$}

We first show that, during the early phases of the algorithm, nodes in $\goodtree$ do not add corrupted information to their $\shortestPath$ variable (see Lemma \ref{lem:shortestPath}).
\begin{lemma} \label{lem:shortestPath}
    Consider any phase $i < \rho$, some iteration $j$, and some $u \in \goodtree$. At the end of iteration $j$, it holds that either $\shortestPath = \texttt{none}$ or $\shortestPath$ corresponds to a shortest path in $G$ starting at some node $v$ that generated a beacon message and ending at $u$.
\end{lemma}
\begin{proof}
    Since $u \in \goodtree$ and $i < \rho$, all Byzantine node are at a distance of at least $i+2$ from $u$, and hence no information injected by a Byzantine node can reach $u$ until it stops waiting for beacon messages in iteration $j$. It follows that any information that was added to $\shortestPath$ corresponds to a path in $G$.
\end{proof}

In Lemma~\ref{lem:borderBlacklisted} we show an upper bound on the number of nodes that are all located at the closest possible distance to $u \in \goodtree$ such that $u$ will blacklist them if they generate beacon messages. We note that some of or even all of such nodes may be good nodes, but that does not cause any conflicts with our argument here. We will use this lemma together with the tree-like property in to argue that the remaining non-blacklisted nodes (and their expanded neighbors) provide a sufficiently large set (see Lemma~\ref{lem:available}) for making it likely that some node generates a beacon message (see Lemma~\ref{lem:expectedNotStoppedi} and Lemma~\ref{lem:whpNotStoppedNodes}).

\begin{lemma} \label{lem:borderBlacklisted}
    Consider any phase $i < \rho$ and some good node $u \in \goodtree$ that has not yet decided at the start of $i$. For each iteration $j$, node $u$ blacklists at most one node in its $\lceil (1-\epsilon)i \rceil$-boundary $\boundary(u, \lceil(1-\epsilon)i\rceil)$ (and none of the nodes that are at a lesser distance).
\end{lemma}
\begin{proof}
    Assume towards a contradiction that, in some iteration $j$, node $u$ adds at least two nodes $w_1, w_2 \in \boundary(u,\lceil(1-\epsilon)i\rceil)$ to its blacklist. By the code of the algorithm, it follows that the ids of $w_1$ and $w_2$ must both be in $\shortestPath$ at the end of iteration $j$. Without loss of generality, suppose that $\shortestPath = (v,\dots,w_1,\dots,w_2,\dots,u)$, i.e., $v$ is the origin of the beacon message that caused the update to $\shortestPath$ in iteration $j$. Since $w_1, w_2 \in \boundary(u,\lceil(1-\epsilon)i\rceil)$ it follows that there exists a path $P_1=(w_1,\dots,u)$ of length $\lceil (1-\epsilon)i \rceil$ between $w_1$ and $u$ that does not contain $w_2$.

    However, this means that $u$ must have received a beacon message containing a path field that contains the concatenation of paths $Q' = (v,\dots,w_1)$ and $P_1$, where $|Q'| < |\shortestPath|$. This contradicts the assumption that both $w_1$ and $w_2$ are in $\shortestPath$, and completes the proof.
\end{proof}

\begin{lemma} \label{lem:available}
    Let $BL_u^*$ denote the set of nodes added to $u$'s blacklist during phase $i < \rho$ and let $A_u^*$ be the set of nodes in $\ball(u,i+2) \setminus BL_u^*$ having a shortest path to $u$ that does not traverse nodes in $BL_u^*$. Then, it holds that $|A_u^*| \ge {d^i}$.
\end{lemma}
\begin{proof}
    Lemma \ref{lem:borderBlacklisted} tells us that during phase $i$, the number of nodes in $\boundary(u,\lceil(1-\epsilon)i\rceil)$ that are added to $BL_u$ is at most
    \begin{align}
        e^{(1-\gamma)i} + 1
        &\le 2e^{(1-\gamma)i} \notag\\
        &\le e^{(1-\gamma)i + \log 2}. \label{eq:boundaryblacklisted}
    \end{align}

    On the other hand,
    \begin{align}
        |\boundary(u,\lceil (1-\epsilon)i \rceil)|
        &\ge d^{(1-\epsilon)i }  \notag\\
        &= e^{(1-\epsilon)i\log d} \notag \\
        &= e^{(1-\delta)\gamma i}. \tag{by \eqref{eq:eps}}
    \end{align}

    This implies that
    \begin{align}
        \frac{|\boundary(u,\lceil (1 - \epsilon)i \rceil)|}{2}
        &\ge e^{(1-\delta)\gamma i - \log 2}, \label{eq:boundary}
    \end{align}

    To show that at most half of the nodes in the $\lceil (1-\epsilon) i \rceil$-boundary of $u$ are blacklisted, it suffices if the right-hand side of \eqref{eq:boundaryblacklisted} is upper bounded by \eqref{eq:boundary}. This is true if
    \begin{center}
        $(1-\gamma)i + \log 2 \le (1-\delta)\gamma i - \log 2$,
    \end{center}
    which holds if
    \begin{equation} \label{eq:gammaDominates}
        \gamma \ge \frac{1}{2-\delta} + \frac{2\log 2}{(2-\delta)i}\text{.}
    \end{equation}

    By the code of the algorithm, we know that $i \ge \tfrac{2\log 2}{(2-\delta)\eta}$  (see Line~\ref{line:startPhase}) and hence \eqref{eq:gammaDominates} is guaranteed by the assumed bound on $\gamma$ stated in \eqref{eq:gamma}.

    So far, we have shown that there is a set $S'$ of at least half of the nodes in the $\lceil (1 - \epsilon)i \rceil$-boundary of $u$ that are \emph{not} in the phase $i$ blacklist $BL_u^*$ of $u$, i.e.,
    \begin{equation} \label{eq:notblacklisted}
        |S'| \ge \frac{\boundary(u,\lceil(1-\epsilon)i\rceil)}{2} \ge d^{(1-\epsilon)i -\log 2 }.
    \end{equation}

    Since $i < \rho$ and $u \in \goodtree$, it follows that each node in $S'$ is the root of a $d$-ary subtree of depth at least $\lfloor \epsilon i \rfloor + 2$. By the tree-like property of $u$, we know that the sets of nodes in these trees are pairwise disjoint. Let $T$ be the set of nodes that are in these trees. By the above,
    \begin{align*}
        |A_u^*| \ge |T| &\ge \frac{\boundary(u,\lceil (1-\epsilon)i\rceil)}{2} \cdot d^{\lfloor \epsilon i\rfloor + 2} \\
        &\ge d^{\lceil (1-\epsilon)i \rceil - \log_d(2) + \lfloor \epsilon i\rfloor + 2} \tag{by \eqref{eq:boundary}} \\
        \intertext{and, assuming $\log_d(2) \le 1$, we get}
        |A_u^*| &\ge d^{\lceil (1-\epsilon)i \rceil + \lfloor \epsilon i\rfloor + 1} \\
        &\ge d^i.
    \end{align*}

    In the remainder of the proof, we argue that none of the nodes in $T$ is blacklisted by $u$.

    Consider any node $w \in T$. The only way that $w$ can be added to $BL_u^*$ is that $w \in \shortestPath$ during some iteration $j$. However, by the tree-like property, we know that any shortest path from $w$ to $u$ must pass through some node $w' \in S'$ and hence, by Lemma \ref{lem:borderBlacklisted}, $\shortestPath$ must contain $w'$. This contradicts the assumption that the nodes in $S'$ are never blacklisted, and thus it follows that none of the nodes in $T$ end up in $u$'s phase $i$ blacklist $BL_u^*$.
\end{proof}
We now show that a large fraction of the nodes in $\goodtree$ do not decide in the first $o(\log n)$ phases.

\begin{lemma} \label{lem:expectedNotStoppedi}
    For any $u \in \goodtree$, $\Pr[u\text{ decides in phase }i]  \leq  \exp(-\frac{c_1i}{2})$.
\end{lemma}
\begin{proof}
    Consider some $u \in \goodtree$ that has not yet decided at the start of phase $i$. By Lemma~\ref{lem:available}, we know there is a set $A_u^*$ of at least $d^i$ nodes such that a beacon message generated by a node in $A_u^*$ is guaranteed to reach $u$ within $(i + 2)$ rounds. By assumption, $u$ has not yet decided at the start of phase $i$, and hence it must have sent out a $\langle continue \rangle$ message at the end of the previous phase $(i - 1)$, which was forwarded for $(i-1)+3 = i+2$ rounds. Given that $A_u^* \subseteq \ball(u,i+2)$, it follows that this beacon message must have reached all nodes in $A_u^*$. Hence, any node in $A_u^*$ that has already decided (possibly during a much earlier phase), will nevertheless execute the for-loop and try to generate beacon messages in the current phase $i$. It follows that $u$ does not decide if at least one node in $A_u^*$ generates a beacon message.

    The probability that none of the nodes in $A_u^*$ becomes active during an iteration of phase $i$ is at most $\lb( 1 - \frac{c_1\ i}{d^i} \rb)^{d^i} \le e^{-c_1 i}$. By taking a union bound over the $e^{(1-\gamma)i}+1$ iterations of phase $i$, it follows that
    \begin{align*}
        \Pr[u\text{ decides in phase }i]
        &\le \lb( e^{(1-\gamma)i} + 1 \rb) e^{-c_1 i } \\
        &=  e^{(1-\gamma - c_1)i}  + e^{-c_1 i } \\
        &\le 2 e^{-(c_1 - 1)i}  &\tag{since $1-\gamma\le 1$}\\
        &\le e^{-(c_1/2)i}\text{.}
    \end{align*}
\end{proof}

Lemma \ref{lem:expectedNotStoppedi} promises us that any individual node has a small probability of error when $i < \rho$. So the expected number of nodes to make an error is also small. We, however, want to show a high probability bound on the number of nodes that make a mistake. In order to show that, we proceed along the usual way of formulating an indicator random variable and then computing the expectation of the sum of the individual indicator random variables by using the principle of linearity of expectation. We show the high probability bound by using the method of \emph{bounded differences} (Azuma's Inequality, more specifically).

Now to the formal description. Let $Y^v_i$ be an indicator random variable which is $1$ if and only if $v$ decides $i$ to be a correct estimate of $\log{n}$. Lemma \ref{lem:expectedNotStoppedi} shows that
$\Pr[Y^v_i = 1]  <  \exp(-\frac{c_1i}{2})$. Now let $$Y_i  =  \sum_{v\in V}{Y^v_i}\text{.}$$ That is, $Y_i$ denotes the number of nodes that decide \emph{wrongly} in the $i^{\text{th}}$ phase. We recall once again that here we are interested only in the case where $i < \rho$. Then $Y_i$ cannot be too large, i.e., not too many nodes can decide wrongly in one phase.
\begin{lemma} \label{lem:whpNotStoppedNodes}
	$\Pr[Y_i   >   2n \cdot \exp(-\frac{c_1i}{2})]   <   \exp(-\frac{\sqrt{n}}{2})$ if $i < \rho$.
\end{lemma}
\begin{proof}
\begin{align}
	&\E[Y_i]  =  \E[\sum_{v\in V}{Y^v_i}]  =  \sum_{v\in V}\E[{Y^v_i}] \tag{by linearity of expectation} \nonumber \\
	&=  \sum_{v\in V}{\Pr[Y^v_i = 1]} \tag{since $Y^v_i$ is an indicator random variable} \nonumber \\
	&<  \sum_{v\in V}\exp(-\frac{c_1i}{2}) \tag{from Lemma \ref{lem:expectedNotStoppedi}} \nonumber \\
	&=  n \cdot \exp(-\frac{c_1i}{2})\text{.} \label{equation-congest-algorithm-Azumas-expected-number-of-wrong-deciders-in-the-lower-end}
\end{align}

Two vertices $v$ and $w$ are \emph{independent} if their $(i + 3)$-distance neighborhoods do not intersect, i.e., if the distance between them is $\geq  2i + 7$. In other words, $v$ going defective can affect only those vertices that are within a distance of $(2i + 6)$ to $v$. In a $(d + 1)$-regular graph, the number of vertices that are within a $(2i + 6)$-distance of $v$ is at most $d^{(2i + 7)}$. By the Azuma-Hoeffding Inequality \cite[Theorem 5.1]{Dubhashi_2009},
\begin{align}
	&\Pr[Y_i - \E[Y_i]   \geq   n \cdot \exp(-\frac{c_1i}{2})] \nonumber  \\
	&\leq   \exp(-\frac{{(n \cdot \exp(-\frac{c_1i}{2}))}^2}{2n \cdot d^{2(2i + 7)}}) \nonumber \\
	&=   \exp(-\frac{n \cdot \exp(-c_1i)}{2d^{4i + 14}}) \nonumber \\
	&=   \exp(-\frac{n}{2 \cdot \exp(c_1i) \cdot d^{4i + 14}}) \nonumber \\
	&=   \exp(-\frac{n}{2e^k})\text{,} \label{equation-congest-algorithm-azumas-first-equation}
\end{align}
say, where
\begin{equation} \label{equation-temporary-variable-k}
    k   \defeq   (4i + 14)\ln{d}  +  c_1i\text{.}
\end{equation}

But we have from Equation \ref{eq:rho} that
\begin{align*}
    &i   <   \rho   <   \frac{\frac{1}{2}\ln{n} - 14\ln{d - 1}}{4\ln{d} + c_1}\\
    &\implies   4i\ln{d} + c_1i   <   \frac{1}{2}\ln{n} - 14\ln{d}\\
    &\text{or, }   (4i + 14)\ln{d}  +  c_1i   <   \frac{1}{2}\ln{n}\\
    &\text{or, }   k   <   \frac{1}{2}\ln{n} \tag{follows from the definition of $k$ in Equation \ref{equation-temporary-variable-k}}
\end{align*}

Substituting the value of $k$ in Equation \ref{equation-congest-algorithm-azumas-first-equation}, we get that
\begin{align}
    &\Pr[Y_i - \E[Y_i]   \geq   n \cdot \exp(-\frac{c_1i}{2})] \nonumber \\
    &<   \exp(-\frac{n}{2\exp(\frac{1}{2}\ln{n})}) \nonumber \\
    &=   \exp(-\frac{n}{2\sqrt{n}})   =   \exp(-\frac{\sqrt{n}}{2})\text{.} \label{equation-congest-algorithm-azumas-final-probability-bound}
\end{align}

Equation \ref{equation-congest-algorithm-azumas-final-probability-bound} combined with Equation \ref{equation-congest-algorithm-Azumas-expected-number-of-wrong-deciders-in-the-lower-end} yields the result.
\end{proof}

As an immediate consequence of Lemma~\ref{lem:whpNotStoppedNodes}, we can take a union bound over the phases up to $i  =  c, c + 1, c + 2, \ldots, \rho$ to obtain that, whp, the total number of nodes in $\goodtree$ that decide early is at most
\begin{equation} \label{eq:earlyDeciders}
    \sum_{i=1}^{\rho} \frac{2n}{e^{(c_1/2)i}}   \le   \frac{2n}{e^{c_1/2 } - 1}.
\end{equation}
\subsubsection{Analysis of the later phases of the algorithm: when $i = \lceil \log{n} \rceil$}

Once a node $u \in \goodtree$ proceeds beyond phase $\rho$, it has obtained a sufficiently good estimate of $\log{n}$, and hence our goal is to show that it is likely to decide. For this part of the analysis, we need to deal with the possibility that conflicting information originating at Byzantine nodes reaches $u$ during an iteration. However, in the following analysis, we show that $u$ is unlikely to increase its phase counter above $\lceil \log n\rceil$.

\begin{lemma} \label{lem:upperEnd}
    Consider phase $i = \lceil \log{n} \rceil$. The following hold with probability at least $1 - O(\frac{1}{n})$:
    \begin{enumerate}
        \item No good node becomes active.
        \item Every node in $\goodtree$ decides at the end of phase $i = \lceil \log n \rceil$.
    \end{enumerate}
\end{lemma}
\begin{proof}
Let $Active(u,j)$ be the event that a good node $u$ becomes active in iteration $j$ of phase $i$. For any $u \in \goodtree$ and any iteration $j$, it holds that
\begin{equation*}
    \Pr[ Active(u,j) ] = \frac{c_1\cdot i}{d^i} \le \frac{c_1\log n}{n^{\log d}}\text{.}
\end{equation*}
By taking a union bound over all good nodes and over all $e^{(1-\gamma)i}+1$ iterations, we get
\begin{align}
    \Pr[ \exists u\colon Active(u,j) ]
    &\le \frac{c_1\log n}{n^{\log(d) - 1}}; \notag\\
    \Pr[\exists j\ \exists u\colon Active(u,j)]
    &\le \lb( e^{\log(n)(1-\gamma)}+1 \rb) \frac{c_1\log n}{n^{\log(d) - 1}} \notag\\
    &\le 2n^{1-\gamma} \frac{c_1\log n}{n^{\log(d) - 1}} \notag\\
    &\le \frac{2c_1\log n}{n^{\log d - 2}}. \notag
\end{align}
This completes the proof of Part (1), assuming $\log{d} \geq 4$.\\

To show Part (2), we \emph{condition} on Part (1) being true, and assume towards a contradiction that there is an undecided node $u \in \goodtree$ that does not decide in phase $i$. We will argue that $u$ blacklists at least one Byzantine node in each iteration $j$ of phase $i$.

By the code of the algorithm, $u$ has set $\shortestPath \gets (v_1, v_2, \ldots, v_k)$, for some $k  \leq  i + 1$, which is the path information of the first beacon message that it received in iteration $j$.

Note that we cannot be sure that $v_1$ is the id of a Byzantine node, as it could have happened that some other Byzantine node $v_\ell$ $(\ell \in [2,k]$) tampered with the prefix $(v_1,\dots,v_{\ell-1})$ before that message reached $u$. However, by Lemma \ref{lem:goodNodes}, we know any Byzantine node is at least $\lfloor {(1-\delta)\gamma}\log_d n\rfloor$ hops away from $u$. In particular, this guarantees that the path suffix $P'$, which consists of the last $\lfloor {(1-\delta)\gamma}\log_d n\rfloor$ nodes in path $P$, contains only ids of good nodes.  Hence at least one Byzantine node's id must be in the path prefix prefix $Q$ (where we obtain $Q$ by removing $P'$ from $P$), as we have assumed that only Byzantine nodes generate beacon messages at this point.

We will now argue that all nodes in $Q$ are blacklisted by $u$. By the description of the algorithm, $u$ blacklists only nodes that have a distance of at least $\lfloor (1-\epsilon) i \rfloor$ from $u$. We observe that
\begin{align*}
    \lfloor (1-\epsilon)i \rfloor   &\leq (1-\epsilon )i   =   {(1-\delta)\gamma}\log_d n. \tag{by \eqref{eq:eps}}
\end{align*}

It follows that the entire prefix $Q$ will be blacklisted. Thus, we have shown that $u$ does not accept a beacon message that visits any of the nodes in $Q$ in a future iteration of this phase.

By the above reasoning, we know that $u$ blacklists at least $1$ Byzantine node in each iteration. Recall that $u$ executes $e^{(1-\gamma)i}+1 \ge n^{1-\gamma}+1$ iterations in phase $i$.  Given that there are only $n^{1-\gamma}$ Byzantine nodes in the network, it follows that there exists an iteration in which $u$ does not set its variable $\shortestPath$ to a value different from \texttt{none} since all Byzantine nodes are already blacklisted at that point. We conclude that $u$ decides in Line \ref{line:decide}, yielding a contradiction.
\end{proof}

\begin{lemma} \label{lem:time}
    At least $(1 - \beta)n$ nodes decide within $O(n^{1-\gamma}\log^2 n)$ rounds of the algorithm.
\end{lemma}
\begin{proof}
    Lemma~\ref{lem:upperEnd}(b) tells us that every node in $\goodtree$ decides by phase $i=\lceil \log n \rceil$ with high probability. By the description of the algorithm, each phase $i$ consists of $2i+3$ rounds and hence the total number of rounds executed until that point is $O(\log^2{n})$.
\end{proof}

We now combine the previous lemmas to complete the proof of Theorem~\ref{thm:congest}.
\begin{proof}[Proof of Theorem \ref{thm:congest}]
    We focus on nodes in $\goodtree$. From Lemma~\ref{lem:whpNotStoppedNodes} we know that $\Omega(n)$ nodes in $\goodtree$ will proceed to at least phase $\rho = \Omega(\log{n})$ before deciding and thus we can set the parameter $\beta$ of the theorem statement accordingly. On the other hand, Lemma \ref{lem:upperEnd} guarantees that all of these nodes decide by the end of phase $\lceil \log n \rceil$ with high probability.

    The claim on the running time follows immediately from Lemma \ref{lem:time}.
\end{proof}

\begin{remark} \label{remark-approximation-factor-for-congest-algorithm-is-not-universal}
    The approximation factor mentioned in Theorem \ref{thm:congest} is not universal. It may be different for different nodes, but in all cases it is bounded by the quantity $\lceil\frac{\log{n}}{\rho}\rceil$, where $\rho$ is as defined in Equation \eqref{eq:rho}. Also, while the estimates may vary by a constant factor, it holds with high probability that all the nodes in $\goodtree$ have estimates that are upper-bounded by $\lceil\log{n}\rceil$, i.e., the estimates of $\log{n}$ are upper-bounded by an additive constant term (which is $1$ basically).
\end{remark}

\begin{remark}
    We point out that a node in $\goodtree$ may reenter the for-loop and participate in sending out beacons even after it has already decided (see Line~\ref{line:reenter}). However, in Corollary~\ref{lem:benign} below, we show that in the benign case where there are no Byzantine nodes in the network, the algorithm computes $\log(n)$ exactly and all nodes terminate.
\end{remark}

\begin{corollary} \label{lem:benign}
    Suppose that all nodes are good. Then the algorithm terminates in $O(\log(n))$ rounds, and whp,  $\Omega(n)$ nodes decide on $\lceil\log(n)\rceil$ and stop sending messages.
\end{corollary}
\begin{proof}
    Lemmas \ref{lem:whpNotStoppedNodes} and \ref{lem:upperEnd} tell us that $\Omega(n)$ nodes will proceed to phase $i = \lceil \log(n) \rceil$. Moreover, none of the good nodes will generate a beacon message with high probability at that point. Thus no node will send a continue message and all nodes will exit the for-loop.
\end{proof}


\section{Impossibility result} \label{sec:impossibility}

We have seen that both of our algorithms crucially rely on the expansion properties of the underlying network. In Theorem \ref{thm:impossibility}, we show that having sufficient expansion is necessary for obtaining \emph{any} approximation of $\log(n)$. In the proof, we make use of the fact that a single Byzantine node can trick the honest nodes into believing that there may be some large number of nodes hidden ``behind'' the Byzantine node, and the honest nodes have no way of verifying whether this bottleneck actually exists.

\begin{theorem} \label{thm:impossibility}
    There is no randomized algorithm that ensures that more than $\lceil\frac{n}{2}\rceil$ nodes output an approximation of $\log{n}$ in the presence of one Byzantine node with probability at least $(1 - \epsilon)$, if there are no restrictions on the network topology of the given $n$-node network, for any constant $0  <  \epsilon  <  1$.
\end{theorem}
\begin{proof}
    For the sake of a contradiction, suppose that there is an algorithm $A$ that solves the counting problem and let $C_n$ be an arbitrary graph of size $n$. Fix any approximation factor $c  =  c(n)  >  0$, which includes $c$ being a constant as a special case. Suppose that an execution $\gamma$ of algorithm $A$ results in a set $S$ of more than $\lceil \frac{n}{2}\rceil$ nodes, where every $u_i \in S$ outputs an estimate $\hat{\ell_i}$ such that $\frac{\ell}{c} \le \hat{\ell_i} \le c \ell$, for some $\ell$. Then we say that \emph{$A$ decides on a common estimate of $\ell$ in execution $\gamma$}.

    For a given $n$, let $t \geq n$ be the smallest integer such that the probability of $A$ deciding on a common estimate of $\log{(nt)}$ when executing on network $C_n$ is at most $\frac{1-\epsilon}{t}$, where $\epsilon>0$ is the assumed constant error probability of the algorithm $A$. Since we assume that decisions are irrevocable (see Def.~\ref{definition-problem-definition}), we know that $A$ can decide at most one common estimate in a single execution, the event of producing common estimate $\ell_1$ and the event of producing common estimate $\ell_2$, where $c\ \ell_1 < \frac{\ell_2}{c}$, are mutually exclusive.
    If $t$ does not exist, then the algorithm has probability $>  \frac{1 - \epsilon}{k}$ to output $\log{(nk)}$ as the common estimate, for all $k \geq n$. However, by summing up the probabilities of these (mutually exclusive) events we get $\sum_{k = n}^{\infty}\frac{1-\epsilon}{k}  >  1$, i.e., the probabilities of outputting the common estimates do not form a valid probability distribution. It follows that $t$ exists.

    Now consider a graph $H$ of $t$ copies of $C_n$ where the Byzantine node $b$ is part of each copy, i.e., node $b$ has degree $t \cdot deg(b)$ where $deg(b)$ is the degree of $b$ in $C_n$. For each copy $C_n$, node $b$ outputs the same set of messages and local state transitions, as are required by an execution of algorithm $A$ in the network $C_n$ for some given random coin flips when $b$ is an honest node. For the algorithm to output a common estimate of $\log(nt)$ when executing on $H$, at least $>\frac{nt}{2}$ nodes need to output a common estimate of $\log{(nt)}$, which involves at least $\frac{n}{2}$ nodes in at least $\frac{t}{2}$ copies of $C_n$. Since the nodes in any given copy of $C_n$ cannot distinguish the execution in $C_n$ from the execution on $H$ and nodes in each individual $C_n$ have probability at most $\frac{1 - \epsilon}{t}$ of outputting the required approximation of $\log(nt)$, we can take a union bound over the $\frac{t}{2}$ copies of $C_n$. Thus, for the probability of the algorithm to produce a common estimate in $H$, we obtain $\frac{1 - \epsilon}{t}\frac{t}{2}  \leq  \frac{1 - \epsilon}{2}$, a contradiction to the assumed probability of success being $\geq  (1 - \epsilon)$.
\end{proof}
\section{Conclusion and Open Problems}

In this paper we take a step towards designing localized, secure, robust, and scalable distributed  algorithms for large-scale networks. We presented two distributed protocols for the fundamental Byzantine counting problem. Our work leaves many questions open. While our deterministic algorithm runs in optimal $O(\log{n})$ rounds, the randomized algorithm takes rounds that is essentially proportional to the number of Byzantine nodes in the network. Thus a main open problem would be to show a polylogarithmic round algorithm for Byzantine counting using small messages or to prove that this is not possible. Another open problem is to show a  algorithm that can tolerate a significantly larger number of Byzantine nodes, e.g., $\Theta(n)$ Byzantine nodes.


\bibliographystyle{plainurl}
\bibliography{references_byzantine_counting_second_paper}

\appendix
\section{Expander Subgraph Lemma}

For completeness, we restate the following lemma; the original proof is in \cite{Augustine_2015_FOCS}. 

\begin{lemma}[cf.\ Lemma 3 in \cite{Augustine_2015_FOCS}] \label{lem:culling}
    Let $G$ be a $n$-node graph with expansion $\phi$ and constant node degree $d$ and suppose that all nodes in a set $F$ are removed from $G$, where $|F| = o(n)$. Then, for any positive constant $c < 1$, there exists a subgraph $H$ of $G$ such that
    \begin{enumerate}
        \item $H$ has expansion $c \phi$, and
        \item $|H|   \geq   n - |F|\left(1 + \frac{1}{\phi(1 - c)}\right)$. 
    \end{enumerate}
\end{lemma}

\begin{proof}
We adapt the proof of Theorem 2.1 in \cite{Bagchi_2006}. Let $G_F$ be the graph yielded by removing the set $F$ from $G$. Perform the following iterative procedure (cf.\ Algorithm {\sf Prune} in \cite{Bagchi_2006}):
\begin{enumerate}
    \item Initialize $G_0 = G_F$. 
    \item In iteration $i\ge 0$, let $S_i$ be any set of up to $|V(G_i)|/2$ nodes with expansion smaller than $c\phi$.
    \item If $S_i$ exists, prune $S_i$ from $G_i$, i.e., $G_{i+1}  = G_i \setminus S_i$.
    \item Let $H$ be graph that we get after the final iteration; note that $H$ has expansion $\ge c\phi$.
\end{enumerate}

We now lower bound the size of $H$. Suppose that the pruning procedure stops after $m$ iterations. For the sake of a contradiction, suppose that $$\frac{|F|}{\phi(1-c)} < \left|\bigcup_{i=0}^m S_i\right|.$$ Define the set $S = \bigcup_{i=0}^{\ell}S_i$ where $\ell$ is the smallest index such that exactly one of the following two cases holds:

First, assume that $|S| \le n/2$. Let $N_{G_i}(S_j)$ denote the neighbors in $G_i \setminus S_j$ of nodes in set $S_j$. We make use of the following result whose proof follows analogously to Lemma 2.6 of \cite{Bagchi_2006}:

\begin{lemma}[cf.\ Lemma 2.6 in \cite{Bagchi_2006}] \label{lem:nbound}
    Suppose that we execute procedure $\sf Prune(c)$. For all $j$ with $0 \leq j < m$, it holds that $| N_{G_F}\left(\bigcup_{i = 0}^j S_i \right) | \le c \phi | \bigcup_{i=0}^j S_i |.$ 
\end{lemma}

Since $G$ has expansion of $\phi$, it holds that $| N_G(S) | \ge \phi |S|$. On the other hand, Lemma \ref{lem:nbound} implies that $| N_{G_F}(S)| \le c\phi |S|$. Thus we get $|F| \ge \phi |S| - c\phi |S|  =  \phi(1 - c)|S|$ and hence $|S|  \leq  \frac{|F|}{\phi(1 - c)}$, yielding a contradiction.

Now, consider the case where $|S| > n/2$. Then, it follows that
\begin{equation} \label{eq:ssize}
    \left|\bigcup_{i = 0}^{\ell - 1}S_i\right| \leq \frac{|F|}{\phi(1 - c)},
\end{equation}
but $|S_\ell| > n/2 - \frac{|F|}{\phi(1-c)}$. (Note that if $S_\ell   \leq   \frac{n}{2} - \frac{|F|}{\phi(1 - c)}$, then $|S|  \leq  \frac{n}{2}$ and the first case applies.) Recalling that $F = o(n)$, it follows that $S_\ell \in \Theta(n)$. Since $S_\ell$ was removed when executing \textsf{Prune($c$)}, we know that $|N_{G_{\ell}}(S_\ell)| \le c\phi |S_\ell|$ and, by the expansion of $G$, we have $|N_G(S_\ell)| \ge \phi |S_\ell|$ as $S_\ell \le n/2$. Thus $$|N_G(S_\ell)| - |N_{G_{\ell}}(S_\ell)| \ge \phi(1 - c)|S_\ell| = \Theta(n)$$ We observe that the size of the neighborhood of $S_\ell$ must have been reduced by $\Theta(n)$ either due to the removal of nodes in $F$ or because of the pruning of the sets $S_0,\dots,S_{\ell-1}$. This, however, yields a contradiction, since
\begin{equation*}
    |F| + \bigcup_{i=0}^{\ell-1}S_i   =   O(|F|)   =   o(n)\text{.}
\end{equation*}
Thus we have shown that
\begin{equation*}
    \left|\bigcup_{i = 0}^m S_i\right| \leq \frac{|F|}{\phi(1 - c)}\text{,}
\end{equation*}
and hence
\begin{equation*}
    |H|   \geq   |G| - |F|(1 + \frac{1}{\phi(1 - c)})\text{,}
\end{equation*}
which completes the proof.
\end{proof}

\end{document}